\documentclass[11pt,a4paper]{scrartcl}
\usepackage[margin=1in]{geometry}

 \pdfoutput=1

\usepackage{caption}

\usepackage[utf8]{inputenc}
\usepackage[section]{placeins}
\usepackage{cite}
\usepackage{subfig}

\usepackage{xcolor}

\usepackage{xspace}
\usepackage{enumerate}

\usepackage{microtype}

\usepackage{amsmath}
\usepackage{amssymb}
\usepackage{amstext}
\usepackage{esvect}
\usepackage{amsthm}
\usepackage{amsopn}
\usepackage{url}
\usepackage[pagebackref]{hyperref}
\usepackage{graphicx}
\usepackage[noline,boxruled,noend,algonl,boxed]{algorithm2e}
\usepackage{paralist}
\usepackage[left, pagewise]{lineno}
\usepackage[textsize=tiny,disable]{todonotes}
\SetVlineSkip{0pt}
\newtheorem{fact}{Fact}  
\newtheorem{lemma}{Lemma}

\newtheorem{theorem}{Theorem}

\newcommand{\dg}{^\circ}

\newcommand{\eps}{\varepsilon}

\renewcommand{\vec}{\vv}

\title{A Note on the Area Requirement\\ of Euclidean Greedy Embeddings\\ of Christmas Cactus Graphs}
 
\author{Roman Prutkin\thanks{{Institute of Theoretical Informatics, Karlsruhe Institute of Technology, Germany}}}  

\date{}
\begin{document}

\maketitle

\begin{abstract}
  An Euclidean greedy embedding of a graph is a straight-line
  embedding in the plane, such that for every pair of vertices~$s$
  and~$t$, the vertex~$s$ has a neighbor~$v$ with smaller distance
  to~$t$ than~$s$. This drawing style is motivated by greedy geometric
  routing in wireless sensor networks.

  A Christmas cactus is a connected graph in which every two simple
  cycles have at most one vertex in common and in which every
  cutvertex is part of at most two biconnected blocks. It has been
  proved that Christmas cactus graphs have an Euclidean greedy
  embedding. This fact has played a crucial role in proving that every
  3-connected planar graph has an Euclidean greedy embedding. The
  proofs construct greedy embeddings of Christmas cactuses of
  exponential size, and it has been an open question whether
  exponential area is necessary in the worst case for greedy
  embeddings of Christmas cactuses. We prove that this is indeed the
  case.~\footnote{This problem has been stated by Ankur Moitra in his
    presentation at the 49th Annual IEEE Symposium on Foundations of
    Computer Science (FOCS'08)~\cite{moitra_results_2008},
    \url{http://people.csail.mit.edu/moitra/docs/ftl.pdf}}

\end{abstract}

\section{Introduction}

Consider a graph~$G=(V,E)$ and a straight-line embedding of~$G$ in the
Euclidean plane. For simplicity, we identify each vertex with the
corresponding point in~$\mathbb R^2$. An embedding of~$G$ is
\emph{greedy} if for every pair~$s,t \in V$, vertex~$s$ has a
neighbor~$v$ in~$G$, for which it is~$|v t| < |s t|$, where~$|p q|$
denotes the Euclidean distance between points~$p$
and~$q$. Equivalently, every pair~$s,t \in V$ is joined by a
distance-decreasing, or \emph{greedy}, path.

Greedy embeddings are motivated by geometric routing in wireless sensor
networks. Given such an embedding, we can use vertex coordinates as
addresses. To route a message to a destination, a vertex can simply
forward the message to a neighbor that is closer to the destination,
and a successful delivery is guaranteed.

The existence of greedy embeddings has been studied for various graph
classes. Papadimitriou and
Ratajczak~\cite{papadimitriou_conjecture_2005} conjectured that every
3-connected planar graph has a greedy embedding in the Euclidean
plane. This conjecture has been proved independently by Leighton and
Moitra~\cite{leighton_results_2009} and Angelini et
al.~\cite{angelini_algorithm_2010}. Both proofs use the fact that
3-connected planar graphs have a spanning \emph{Christmas cactus}
subgraph. A Christmas cactus is a connected graph in which every two
simple cycles have at most one vertex in common and in which every
cutvertex is part of at most two biconnected blocks. The authors show
that every Christmas cactus has a greedy embedding. However, both
constructions produce embeddings of exponential size in the worst
case.

In order for the greedy embedding to be practical for geometric
routing, it must be possible to represent vertex coordinates using
only few bits, otherwise, message headers containing the destination
address would be too big~\cite{eppstein_succinct_2011}. Goodrich and
Strash~\cite{goodrich_succinct_2009} showed how to construct an
Euclidean greedy embedding of a Christmas cactus, in which the
coordinates of every vertex can be encoded using only~$O(\log n)$
bits. In the presented encoding scheme for the vertex coordinates,
their positions in the Euclidean plane are not stored explicitly, and
the drawings might still have exponential size. Angelini et
al.~\cite{angelini_succinct_2012} proved that some trees require
exponential aspect ratio of the edge lengths. It was open whether this
bound also holds for Christmas cactuses. In this note we prove
Moitra's conjecture that Euclidean greedy embeddings of Christmas
cactuses require exponential area in the worst
case~\cite{leighton_results_2009}.

\begin{figure}[tb]
\hfill
  \subfloat[]{\includegraphics[page=1]{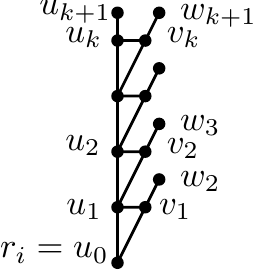} \label{fig:Gk}}
\hfill
  \subfloat[]{\includegraphics[page=2]{fig/exponential} \label{fig:Fk}}
\hfill\null  
\caption{Family of Christmas cactuses that requires exponential area
  for every greedy embedding. \protect\subref{fig:Gk}~Cactus~$G_k$
  for~$k=4$; \protect\subref{fig:Fk}~cactus~$F_k$ constructed by
  attaching the roots of 30~copies of~$G_k$ to a cycle of size~31.}
\end{figure}

\section{Exponential worst case resolution}

We now present a family of Christmas cactuses that requires
exponential aspect ratio of edge lengths in every greedy
embedding. For an integer~$k \geq 1$, consider the Christmas
cactus~$G_k$ with root~$r_i$ in Fig.~\ref{fig:Gk}. We then construct
the cactus~$F_k$ by attaching the roots of 30~copies of~$G_k$ to a
cycle of size~31; see Fig.~\ref{fig:Fk}. We shall prove that the
aspect ratio of edge lengths in every greedy embedding of~$F_k$ is at
least~$2^k$. The following fact follows from Lemma~3
in~\cite{nollenburg_self-approaching_2016}.

\begin{fact}
  Every greedy embedding of~$F_k$ contains a greedy embedding of~$G_k$, in
  which every pair of vectors
  from~$\bigcup_i \{ \vv{u_i u_{i+1}}, \vv{u_i v_{i+1}}, \vv{v_i
    w_{i+1}} \}$ forms an angle of less than~$12\dg$.
  \label{lem:narrow-Gk}
\end{fact}

From now on, we consider the embedding of~$G_k$ from
Fact~\ref{lem:narrow-Gk}.

\begin{lemma}
  For~$0 \leq i \leq k-1$, it holds:
  $|u_{i+1} u_{i+2}| < \frac{1}{2} |u_i u_{i+1}|$.
 \label{lem:lengths-decrease}
\end{lemma}

\begin{figure}[tb]
\centering
\includegraphics[page=3]{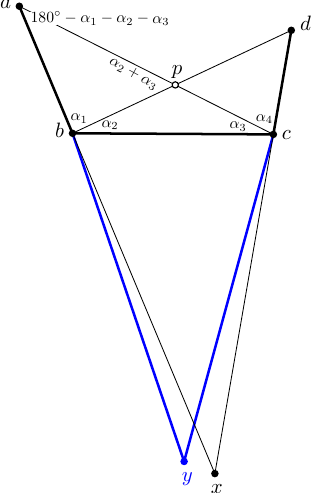} 
\caption{Proof of Lemma~\ref{lem:lengths-decrease}.}
\label{fig:lengths-decrease}
\end{figure}

\begin{proof}
  We rename the vertices for brevity: $a = u_{i+2}$, $b = u_{i+1}$,
  $c = v_{i+1}$, $d = w_{i+2}$, $y = u_i$; see
  Fig.~\ref{fig:lengths-decrease}.  Note that every greedy $a$-$d$
  path as well as every greedy $d$-$a$ path must contain~$b$
  and~$c$. Therefore, the path~$a b c d$ is greedy in both
  directions. Thus, the ray with origin~$b$ and direction~$\vec{b a}$
  and the ray with origin~$c$ and direction~$\vec{c d}$
  diverge~\cite{angelini_succinct_2012}. The paths~$a b d$ and~$a c d$
  are also greedy in both directions, therefore,
  $\alpha_1 = \angle a b d > 60\dg$
  and~$\alpha_4 = \angle a c d > 60\dg$.

  Let~$x$ be the intersection point of the lines through~$a b$
  and~$c d$. Let~$\eps = 12\dg$. Since~$G_k$ has been chosen according
  to Fact~\ref{lem:narrow-Gk}, it is $\angle x b y < \eps$
  and~$\angle x c y < \eps$.

  It is~$\angle c b x = 180\dg - \angle a b c < 120 \dg$. Similarly,
  $\angle b c x < 120\dg$. Also, $\angle b x c < \eps$. Thus, by
  considering the triangle~$b c x$ it
  follows:~$\angle c b x > 60\dg - \eps$
  and~$\angle b c x > 60\dg - \eps$. Since it
  is~$60\dg - \eps < \angle c b x < 120\dg$, it
  is~$60\dg - 2 \eps < \angle c b y < 120\dg + \eps$. Analogously, it
  is~$60\dg - 2 \eps < \angle b c y < 120\dg + \eps$. It follows:

  \begin{equation*}
    \frac{|b c|}{|b y|} = \frac{\sin \angle b y c}{\sin \angle b c y} < \frac{\sin \eps }{\sin(60\dg - 2 \eps)} < 0.36.
  \end{equation*}

  Therefore, it is~$|b c| < 0.36 |b y|$ and, analogously,
  $|b c| < 0.36 |c y|$.

  Next, recall that it
  is~$\angle b x c = \alpha_1 + \alpha_2 + \alpha_3 + \alpha_4 -
  180\dg < \eps$, for~$\alpha_2 = \angle d b c$
  and~$\alpha_3 = \angle a c b$.
  Therefore,~$\angle b a c = 180\dg - \alpha_1 - \alpha_2 - \alpha_3 >
  \alpha_4 - \eps > 60\dg - \eps$. Also, since the path~$a b c$ is
  greedy in both directions, it is~$\angle b a c < 90\dg$. Now
  consider~$\angle a c b = \alpha_3$.
  Since~$\angle b c x > 60\dg - \eps$, it is
  $\alpha_3 + \alpha_4 < 120\dg + \eps$,
  and~$\alpha_3 < 60\dg + \eps$. Therefore,

  \begin{equation*}
    \frac{|a b|}{|b c|} = \frac{\sin \angle a c b}{\sin \angle b a c} = \frac{\sin \alpha_3}{\sin (180\dg - \alpha_1 - \alpha_2 - \alpha_3)} < \frac{\sin (60\dg + \eps) }{\sin(60\dg - \eps)} < 1.28.
  \end{equation*}

  Thus, $|a b| < 1.28 |b c|$. It follows:
  $|a b| < 1.28 |b c| < 1.28 \cdot 0.36 |b y| < 0.461 |b
  y|$. Therefore, we have
  $|u_{i+1} u_{i+2}| < \frac{1}{2}|u_i u_{i+1}|$.
\end{proof}

\begin{theorem}
  In every greedy embedding of cactus~$F_k$, the ratio of the longest
  and the shortest edge is in~$\Omega(2^{n/90})$, where~$n$ is the
  number of vertices of~$F_k$.
  \label{thm:exponential}
\end{theorem}

\begin{proof}
  Cactus~$G_k$ has~$3 k + 2$ vertices. Thus, cactus~$F_k$
  has~$n = 90 k + 61$ vertices. By Lemma~\ref{lem:lengths-decrease},
  every greedy embedding of~$F_k$ contains an embedding of~$G_k$, such
  that it is~$|u_k u_{k+1}| < \frac{1}{2^k}|u_0 u_1|$. Therefore, the
  ratio of the longest and shortest edge in every greedy embedding
  of~$F_k$ is at least~$2^k = \Omega(2^{n/90})$.
\end{proof}

\subsection*{Acknowledgements}

The author thanks Martin Nöllenburg for valuable discussions and
comments.


\begin{thebibliography}{1}

\bibitem{angelini_succinct_2012}
P.~Angelini, G.~{Di Battista}, and F.~Frati.
\newblock Succinct greedy drawings do not always exist.
\newblock {\em Networks}, 59(3):267--274, 2012.

\bibitem{angelini_algorithm_2010}
P.~Angelini, F.~Frati, and L.~Grilli.
\newblock An algorithm to construct greedy drawings of triangulations.
\newblock {\em J. Graph Algorithms Appl.}, 14(1):19--51, 2010.

\bibitem{eppstein_succinct_2011}
D.~Eppstein and M.~T. Goodrich.
\newblock Succinct greedy geometric routing using hyperbolic geometry.
\newblock {\em IEEE Transactions on Computers}, 60(11):1571--1580, 2011.

\bibitem{goodrich_succinct_2009}
M.~T. Goodrich and D.~Strash.
\newblock Succinct greedy geometric routing in the {{Euclidean}} plane.
\newblock In Y.~Dong, D.-Z. Du, and O.~Ibarra, editors, {\em Algorithms and
  Computation (ISAAC'09)}, volume 5878 of {\em LNCS}, pages 781--791. Springer,
  2009.

\bibitem{leighton_results_2009}
T.~Leighton and A.~Moitra.
\newblock Some results on greedy embeddings in metric spaces.
\newblock {\em Discrete Comput. Geom.}, 44(3):686--705, 2009.

\bibitem{moitra_results_2008}
A.~Moitra and T.~Leighton.
\newblock Some results on greedy embeddings in metric spaces.
\newblock In {\em Foundations of Computer Science (FOCS'08)}, pages 337--346,
  2008.

\bibitem{nollenburg_self-approaching_2016}
M.~N{\"o}llenburg, R.~Prutkin, and I.~Rutter.
\newblock On self-approaching and increasing-chord drawings of 3-connected
  planar graphs.
\newblock {\em J. Comput. Geom.}, 7(1):47--69, 2016.

\bibitem{papadimitriou_conjecture_2005}
C.~H. Papadimitriou and D.~Ratajczak.
\newblock On a conjecture related to geometric routing.
\newblock {\em Theoret. Comput. Sci.}, 344(1):3--14, 2005.

\end{thebibliography}

\end{document}